\documentclass[letterpaper,USenglish,numberwithinsect]{lipics-v2018-ascii-180529-adapted-so-works-for-tr-so-no-yellow-strips-and-no-page-one-confinfo}

\usepackage[protrusion=false]{microtype}

\usepackage{tikz}
\usepackage{pgfplots,wrapfig}
\usepackage{floatflt}

\pgfplotsset{compat=1.13}  

\usepackage[hang]{footmisc}

\nolinenumbers
\DeclareMathOperator{\sharpp}{\#P}
\newcommand{\p}{{\rm P}}
\DeclareMathOperator{\up}{UP}
\DeclareMathOperator{\coup}{coUP}
\newcommand{\psharpp}{\p^{\sharpp}}

\def\hackspace{{\hspace{0.75pt}}}
\newcommand{\spr}{strong-\allowbreak P\trank}
\newcommand{\pr}{\text{semistrong-}\allowbreak P\trank}
\newcommand{\wpr}{P\trank}

\newcommand{\nspr}{strong-\allowbreak P\trank\ensuremath{^\complement}}
\newcommand{\nsprit}{strong-\allowbreak P\trank\ensuremath{\hackspace^\complement}}

\newcommand{\npr}{\text{semistrong-}\allowbreak P\trank\ensuremath{^\complement}}
\newcommand{\nprit}{\text{semistrong-}\allowbreak P\trank\ensuremath{\hackspace^\complement}}

\newcommand{\nwpr}{P\trank\ensuremath{^\complement}}
\newcommand{\nwprit}{P\trank\ensuremath{\hackspace^\complement}}

\newcommand{\INFINITE}{\ensuremath{\mathrm{INFINITE}}}
\newcommand{\pc}{P\tcomp\ensuremath{'}}

\newcommand{\frr}{\ensuremath{\frec}\trank}

\newcommand{\frc}{\ensuremath{\frec}\tcomp}
\newcommand{\fprc}{\ensuremath{\fre}\tcomp}

\DeclareMathOperator{\shift}{shift}

\newcommand{\frec}{\ensuremath{\mathrm{F}_\mathrm{REC}}}
\newcommand{\fre}{\ensuremath{\mathrm{F_\mathrm{PR}}}}

\newcommand{\rank}{\ensuremath{\mathrm{rank}}}
\def\tcomp{\text{-compressible}}
\def\trank{\text{-rankable}}
\newcommand{\condition}{\mid}

\newcommand{\sigmastar}{{\ensuremath{\Sigma^\ast}}}

\newcommand{\domain}{\ensuremath{\mathrm{domain}}}
\newcommand{\calf}{\ensuremath{{\cal F}}}
\newcommand{\calc}{\ensuremath{{\cal C}}}
\newcommand{\N}{\mathbb{N}}
\newcommand{\Z}{\mathbb{Z}}

\newcommand{\emptystring}{\epsilon}

\newcommand{\rec}{\ensuremath{\mathrm{REC}}}
\newcommand{\re}{\ensuremath{\mathrm{RE}}}
\newcommand{\core}{\ensuremath{\mathrm{coRE}}}

\newcommand{\join}{\oplus}

\theoremstyle{plain}
\newtheorem{proposition}[theorem]{Proposition}

\hideLIPIcs

\title{Closure and Nonclosure Properties of the Compressible and Rankable Sets}

\author{Jackson Abascal\footnote{Supported in            
part by a CRA-W Collaborative Research Experiences for Undergraduates (CREU) grant.}, 
Lane A. Hemaspaandra\footnote{This work was done in part
     while on a sabbatical stay at ETH Z\"{u}rich and the University of D\"usseldorf.},
Shir Maimon$^{\rm{}1,}$\footnote{Current affiliation: Department of Computer Science, Cornell University},
and Daniel Rubery}
{Department of Computer Science, 
University of Rochester, Rochester,
   NY 14627, USA\\~\\November 5, 2016; revised October 30, 2018}{}{}{}

\authorrunning{J. Abascal, L.\,A. Hemaspaandra, S. Maimon, and
D. Rubery}

\Copyright{Jackson Abascal, Lane A. Hemaspaandra, Shir Maimon, and Daniel
Rubery}

\subjclass{
	\ccsdesc[500]{Theory of computation~Computational complexity and cryptography};
	\ccsdesc[300]{Theory of computation~Complexity classes};
	\ccsdesc[300]{Theory of computation~Computability};
	\ccsdesc[100]{Information systems~Data compression}
}

\keywords{complexity theory, closure properties,
compression, ranking,
computability}

\EventEditors{John Q. Open and Joan R. Access}
\EventNoEds{2}
\EventLongTitle{42nd Conference on Very Important Topics (CVIT 2016)}
\EventShortTitle{CVIT 2016}
\EventAcronym{CVIT}
\EventYear{2016}
\EventDate{December 24--27, 2016}
\EventLocation{Little Whinging, United Kingdom}
\EventLogo{}
\SeriesVolume{42}
\ArticleNo{}

\newcount\hour  \newcount\minutes  \hour=\time  \divide\hour by 60
\minutes=\hour  \multiply\minutes by -60  \advance\minutes by \time
\def\mmmddyyyy{\ifcase\month\or Jan\or Feb\or Mar\or Apr\or May\or Jun\or Jul\or Aug\or Sep\or Oct\or Nov\or Dec\fi \space\number\day, \number\year}
\def\hhmm{\ifnum\hour<10 0\fi\number\hour :
\ifnum\minutes<10 0\fi\number\minutes}

\hideLIPIcs

\begin{document}
\sloppy
\maketitle

\begin{abstract}

The rankable and compressible sets have been studied for more than a
quarter of a century, ever since
Allender~\cite{all:thesis:invertible} and Goldberg and
Sipser~\cite{gol-sip:j:compression} introduced the formal study of
polynomial-time ranking.  Yet even after all that time,
whether the rankable and compressible sets are closed under the most
important boolean and other operations remains essentially
unexplored.  The present paper
studies these questions for both polynomial-time and
recursion-theoretic compression and ranking, and for almost every
case arrives at a Closed, a Not-Closed, or a
Closed-Iff-Well-Known-Complexity-Classes-Collapse result for the
given operation.  Even though compression and ranking classes are
capturing something quite natural about the structure of sets, it
turns out that they are
quite fragile with respect to closure properties, and many fail to
possess even the most basic of closure properties.  For example, we
show that with respect to the join (aka~disjoint union) operation:
the P-rankable sets are not closed, whether the semistrongly
P-rankable sets are closed is closely linked to whether
$\p = \up \cap \coup$, and the strongly P-rankable sets are closed.

\end{abstract}

\section{Introduction}

Loosely put, a compression function $f$ for a set $A$ is a function
over the domain $\sigmastar$ such that (a)~$f(A) = \sigmastar$ and
(b)~$(\forall a, b \in A : a \neq b)[f(a)\neq f(b)]$.  That
is, $f$ puts $A$ in 1-to-1 correspondence with $\sigmastar$.  This is
sometimes described as providing a minimal perfect hash function for
$A$: It is perfect since there are no collisions (among elements of
$A$), and it is minimal since not a single element of the codomain is
missed.  Note that the above does not put any constraints on what
strings the elements of $\overline{A}$ are mapped to, or even about
whether the compression function needs to be defined on such
strings. A ranking function is similar, yet stronger, in that a
ranking function sends the $i$th string in $A$ to the integer $i$; it
respects the ordering of the members of $A$.

The study of ranking was started by
Allender~\cite{all:thesis:invertible} and Goldberg and
Sipser~\cite{gol-sip:j:compression}, and has been pursued in many
papers since, especially in the early 1990s, e.g., \cite{%
hem-rud:j:ranking,%
huy:j:rank,%
ber-gol-sab:j:rank,%
alv-jen:j:logcount%
}.  The study of ranking led to the study of compression, which was
started---in its current form, though already foreshadowed in a notion
of~\cite{gol-sip:j:compression}---by Goldsmith, Hemachandra, and
Kunen~\cite{gol-hem-kun:j:address} (see
also~\cite{gol-hom:j:scalability}).
The abovementioned work focused on polynomial-time or
logarithmic-space ranking or compression functions.  More recently,
both compression and ranking have also been studied
in the recursion-theoretic context
(\cite{hem-rub:jtoappear-with-tr-pointer:rft-compression1}, and see the discussion therein
for precursors in classic recursive function theory), in particular
for both the case of (total) recursive compression/ranking functions
(which of course must be defined on all inputs in $\sigmastar$) and
the case of partial-recursive compression/ranking functions (i.e.,
functions that on some or all elements of the complement of the set
being compressed/ranked
are allowed to be undefined).

\begin{table}[!b]
{
     \footnotesize  
\begin{center}
\resizebox{\textwidth}{!}
{\begin{tabular}{l | c c c}
Class & $\cap$ & $\cup$ &
complement\\
\hline
\vspace{-4pt}
\\
\spr{}     & $\p=\psharpp$ (Th.~\ref{thm11})& $\p=\psharpp$ (Th.~\ref{thm11}) & Yes (Prop.~\ref{prop-was-lem-lem13}) \\[2pt]
\pr{}      & $\p=\psharpp$ (Th.~\ref{thm11}) & $\p=\psharpp$ (Th.~\ref{thm11})& $\approx \p = \up \cap \coup $ (Th.~\ref{thm16}, Cor~\ref{c:new-cite-for-old-thm18}) \\[4pt]
\parbox{1.7in}{\wpr{},
\pc{},
$\frec$-rank\-able,
$\frec$-com\-press\-ible,
$\fre$-rank\-able, and
$\fre$-com\-press\-ible}
& No (Th.~\ref{yellow})& No (Th.~\ref{red})& No (Th.~\ref{blue})\\[22pt]

\nspr{} & No (Th.~\ref{green})& No (Th.~\ref{green})& Yes (Prop.~\ref{prop-was-lem-lem13})\\[2pt]

\npr  & No (Th.~\ref{green})& No (Th.~\ref{green})& $\approx \p = \up \cap \coup$ (Th.~\ref{thm16}, Cor~\ref{c:new-cite-for-old-thm18})\\[4pt]
\parbox{1.7in}{\wpr{}${}^\complement$,
P-com\-press\-ible${}^\complement$,
$\frec$-rank\-able${}^\complement$,
$\frec$-com\-press\-ible${}^\complement$,
$\fre$-rank\-able${}^\complement$, and
$\fre$-com\-press\-ible${}^\complement$}
& No (Th.~\ref{green}) & No (Th.~\ref{green})& No (Th.~\ref{blue})\\
\end{tabular}
} \vspace*{14pt}

\caption{Overview of results for closure of these classes under
boolean operations.
If an entry does
not contain ``No'' or ``Yes'' then the class is closed under the
operation if and only if the entry holds. A
special case is
\pr{} and \npr{}, in which we deliberately use the
$\approx$ symbol to indicate that the implication is true in one
direction and in the other direction currently is known to be true only
for a broad subclass of
these sets. Specifically, if $\p=\up\cap\coup$ then the
complements of all ``nongappy'' \pr{}
sets are themselves \pr.
}\label{table}
\end{center}
} 
\end{table}

In the present paper, we continue the study of both
complexity-theoretic and recursion-theoretic compression and ranking
functions.  In particular, the earlier papers often viewed the
compressible sets or the rankable sets as a \emph{class}.  We take
that very much to heart, and seek to learn whether these classes do,
or do not, possess key closure properties.  Our main contributions can
be seen in Table~\ref{table}, where we obtain closure and nonclosure
results for many previously studied variations of compressible and
rankable sets under boolean operations (Section~\ref{s:boolean}).  We
also study the closure of these sets under additional operations, such
as the join, aka disjoint union (Section~\ref{s:join}).  And we
introduce the notion of compression \textit{onto a set} and
characterize the robustness of compression under this notion. In
particular, by a finite-injury priority argument with some interesting
features we show that there exist RE sets that each compress to the
other, yet that nonetheless are not
recursively isomorphic (Section~\ref{s:iso}).

\label{s:definitions}
\section{Definitions}

Throughout this paper, ``$\p$'' when used in a function context (e.g.,
the P-rankable sets) will denote the class of total, polynomial-time
computable functions from $\sigmastar$ to $\sigmastar$.
Additionally, throughout this paper,
$\Sigma = \{0,1\}$.
$\frec$ will denote the class of
total, recursive functions from $\sigmastar$ to $\sigmastar$. $\fre$
will denote the class of
partial recursive functions from $\sigmastar$ to $\sigmastar$.
$\emptystring$ will
denote the empty string. We define the function $\shift(x,n)$ for
$n \in \Z$. If $n \geq 0$, then $\shift(x,n)$ is the string $n$ spots
after $x$ in lexicographical order, e.g.,
$\shift(\emptystring,4)= 01$.  For $n > 0$, define $\shift(x,-n)$ as
the string $n$ spots before $x$ in lexicographical order, or
$\emptystring$ if no such string exists. We define the symmetric
difference $A \bigtriangleup B = (A-B) \cup (B-A)$. The symbol $\N$
will denote the natural numbers $\{0,1,2,3,\dots\}$.

We now define the notions of compressible and rankable sets.
\begin{definition}[Compressible
sets~\cite{hem-rub:jtoappear-with-tr-pointer:rft-compression1}]\label{d:compressible}
~\begin{enumerate}
\item\label{d:compressible:p1} Given a set $A\subseteq \sigmastar$, a (possibly partial) function $f$ is a
\emph{compression function for $A$} exactly if
\begin{enumerate}
\item $\domain(f) \supseteq A$,
\item $f(A) = \sigmastar$, and
\item for all $a$ and $b$ in $A$, if $a \neq b$ then
$f(a) \neq f(b)$.
\end{enumerate}
\item\label{d:compressible:p2} Let $\calf$ be any class of (possibly
partial) functions mapping from $\sigmastar$ to $\sigmastar$.  A
set $A$ is \emph{$\calf$-compressible} if some $f \in \calf$ is a
compression function for A.
\item\label{d:compressible:p3} For each $\calf$ as above,
$\calf\text{-compressible} = \{ A \condition A \text{ is }
\calf\text{-compressible}\}$ and
$\calf\text{-compressible}' = \allowbreak
\calf\text{-compressible} \cup \{A \subseteq \sigmastar\condition
A \text{ is a finite set}\}$.
\item\label{d:compressible:p4} For each $\calf$ as above and each
$\calc \subseteq 2^{\sigmastar}$, we say that $\calc$ is
\emph{$\calf$-compressible} if all infinite sets in $\calc$ are
$\calf$-compressible.
\end{enumerate}
\end{definition}

Note that a compression function $f$ for $A$ can have any behavior on
elements of $\overline{A}$ and need not even be defined. Finite sets
cannot have compression functions as they do not have enough elements
to be mapped onto $\Sigma^*$. Thus part 4 of
Definition~\ref{d:compressible:p2} defines a class to be
$\calf$-compressible if and only if its \emph{infinite} sets are
$\calf$-compressible.

Ranking can be informally
thought of as a sibling of compression that preserves lexicographical
order within the set.
We consider three classes of rankable functions that differ in how
they are allowed to behave on the complement of the set they rank.
Although ever since the paper of
Hemachandra and Rudich~\cite{hem-rud:j:ranking},
which introduced two of the three types, there have been those three
types of ranking classes, different papers have used different (and
sometimes conflicting) terminology for these types.  Here, we use the
(without modifying adjective) terms ``ranking function'' and ``rankable''
in the same way as Hemaspaandra and
Rubery~\cite{hem-rub:jtoappear-with-tr-pointer:rft-compression1} do, for the least
restrictive form of ranking (the one that can even ``lie'' on the
complement).  That is the form of ranking that is most naturally
analogous with compression, and so it is natural that both terms
should lack a modifying adjective.  For the most restrictive form of
ranking, which even for strings $x$ in the complement of the set $A$
being ranked must determine the number of strings up to $x$ that are
in $A$, like Hemachandra and Rudich~\cite{hem-rud:j:ranking} we use
the terms ``strong ranking function'' and ``strong(ly) rankable.''  And
for the version of ranking that falls between those two, since for
strings in the complement it need only detect that they are in the
complement, we use the terms ``semistrong ranking function'' and
``semistrong(ly) rankable.''

\begin{definition}[\cite{all:thesis:invertible,gol-sip:j:compression}]
$\rank_A(y) = \| \{ z \condition z \leq y\, \land \, z \in A \}\|$.
\end{definition}

\begin{definition}[Rankable sets,
\cite{all:thesis:invertible,gol-sip:j:compression}, see also
\cite{hem-rub:jtoappear-with-tr-pointer:rft-compression1}]\label{d:rankable}~\begin{enumerate}
\item\label{part:rfn} Given a set $A\subseteq \sigmastar$, a (possibly partial) function $f$ is a \emph{ranking
function for $A$} exactly if
\begin{enumerate}
\item $\domain(f) \supseteq A$ and
\item if $x \in A$, then $f(x) = \rank_A(x)$.

\end{enumerate}
\item Let $\calf$ be any class of (possibly partial) functions
mapping from $\sigmastar$ to $\sigmastar$.  A set $A$ is
\emph{$\calf$\trank} if some $f \in \calf$ is a ranking function
for $A$.
\item For each $\calf$ as above, $\calf$\trank =
$\{ A \condition A \text{ is } \calf\text{-rankable}\}$.
\item\label{d:rankable:part-class} For each $\calf$ as above and
each $\calc \subseteq 2^{\sigmastar}$, $\calc$ is
\emph{$\calf$-rankable} if all sets in $\calc$ are $\calf$\trank.
\end{enumerate}
\end{definition}

\begin{definition}[Semistrongly rankable sets,
\cite{hem-rud:j:ranking}, see also
\cite{hem-rub:jtoappear-with-tr-pointer:rft-compression1}]\label{d:ssrankable}~\begin{enumerate}
\item Given a set $A\subseteq \sigmastar$, a function $f$ is a \emph{semistrong ranking function
for $A$} exactly if
\begin{enumerate}
\item $\domain(f) =\Sigma^*$,
\item if $x \in A$, then $f(x) = \rank_A(x)$, and
\item if $x \notin A$, $f(x)$ indicates ``not in set'' (e.g., via
the machine computing $f$ halting in a special state; we still
view this as a case where $x$ belongs to $\domain(f)$).
\end{enumerate}
\item Let $\calf$ be any class of functions mapping from
$\sigmastar$ to $\sigmastar$.  A set $A$ is \emph{semistrong-$\calf$-rankable} if some $f \in \calf$ is a
semistrong ranking function for $A$.
\item For each $\calf$ as above,
semistrong-$\calf\text{-rankable} = \{ A \condition A \text{ is
semistrong-} \calf\text{-rankable}\}$.
\item For each $\calf$ as above and each
$\calc \subseteq 2^{\sigmastar}$, we say that $\calc$ is
\emph{semistrong-$\calf$-rankable} if all sets in $\calc$ are
semistrong-$\calf$-rankable.
\end{enumerate}
\end{definition}

\begin{definition}[Strongly rankable sets, \cite{hem-rud:j:ranking},
see also
\cite{hem-rub:jtoappear-with-tr-pointer:rft-compression1}]\label{d:srankable}~\begin{enumerate}
\item Given a set $A\subseteq \sigmastar$,
a function $f$ is a \emph{strong ranking function for
$A$} exactly if
\begin{enumerate}
\item $\domain(f) =\Sigma^*$ and
\item $f(x) = \rank_A(x)$.
\end{enumerate}
\item Let $\calf$ be any class of functions mapping from
$\sigmastar$ to $\sigmastar$.  A set $A$ is
\emph{strong-$\calf$-rankable} exactly if
$(\exists f \in \calf)[f \text{ is a strong ranking function for }
A]$.
\item For each $\calf$ as above,
strong-$\calf\text{-rankable} = \{ A \condition A \text{ is
strong-} \calf\text{-rankable}\}$.
\item For each $\calf$ as above and each
$\calc \subseteq 2^{\sigmastar}$, we say that $\calc$ is
\emph{strong-$\calf$-rankable} if all sets in $\calc$ are
strong-$\calf$-rankable.
\end{enumerate}
\end{definition}

For almost any natural class of functions,
$\calf$, we will have that $\calf$\trank{} is contained in
$\calf$\tcomp{}$'$.  In particular, $\p$, $\fre$, and $\frec$ each
have this property.
If $f$ is a ranking function for $A$ (in the sense of
part~\ref{part:rfn} of Definition~\ref{d:rankable}), for our
same-class compression function for $A$ we can map $x \in \Sigma^*$ to
the
$f(x)$-th string in
$\Sigma^*$ (where we consider $\epsilon$ to be the first string in
$\Sigma^*$) if $f(x) > 0$, and if $f(x) = 0$ what we map to is
irrelevant so map to any particular fixed string (for concreteness,
$\epsilon$).

For each class
$\calc \subseteq 2^{\sigmastar}$,
$\calc^{\complement}$ will denote
the complement of $\calc$, i.e., $2^{\sigmastar} - \calc$.  For
example, \nwpr{} is the class of non-\wpr{} sets.

The class \pr{} is a subset of $\p$ (indeed, 
a strict 
subset unless $\p=\p^{\sharpp}$~\cite{hem-rud:j:ranking}), 
but there exist undecidable sets that are \wpr.
Clearly, 
the class of semistrong-\rec\trank{} sets equals
the class of strong-\rec\trank{} sets.

\section{\boldmath{}Compression onto \texorpdfstring{$B$}{B}:
Robustness with Respect to Target Set}\label{s:iso}

A compression function for a set $A$ is 1-to-1 and onto $\sigmastar$
when
the function's domain is restricted to $A$.
It is natural to wonder what changes when we switch target sets from
$\sigmastar$ to some other set $B \subseteq \sigmastar$. We now define
this notion.  In our definition, we do allow strings in $\overline{A}$
to be mapped to $B$ or to $\overline{B}$, or even, for the case of
$\fre$ maps,
to be undefined.  
In particular, this definition does not require that $f(\sigmastar) = B$.
Recall from Section~\ref{s:definitions} that, throughout this paper, 
$\Sigma = \{0,1\}$.

\begin{definition}[Compressible to
$B$]\label{d:compressible-to-B-175.1-dan}~
\begin{enumerate}
\item Given sets $A\subseteq \sigmastar$ and
$B \subseteq \sigmastar$, a (possibly partial)
function $f$ is a \emph{compression function
for $A$ to $B$}
exactly if
\begin{enumerate}
\item $\domain(f) \supseteq A$,
\item $f(A) = B$, and

\item for all $a$ and $b$ in $A$, if $a \neq b$ then
$f(a) \neq f(b)$.
\end{enumerate}
\item Let $\calf$ be any class of (possibly partial) functions
mapping from $\sigmastar$ to $\sigmastar$.  A set $A$ is
\emph{$\calf$-compressible to $B$} if some $f \in \calf$ is a
compression function for $A$ to $B$.
\end{enumerate}
\end{definition}

The classes $\calf$ of interest to us will be $\frec$ and $\fre$.
Clearly, compression is simply the $B=\sigmastar$ case of this
definition, e.g.,
a function $f$ is a compression function for $A$ if and only if $f$ is
a compression function for $A$ to $\sigmastar$, and set $A$ is
$\calf$\tcomp\ if and only if $A$ is $\calf$\tcomp\ to $\sigmastar$.

A natural first question to ask is whether compression to $B$ is a new
notion, or whether it coincides with our existing notion of
compression to $\sigmastar$, at least for sets $B$ from common classes
such as $\rec$ and $\re$. The following result shows that for $\rec$
and $\re$ this new notion does coincide with our existing one.

\begin{theorem}\label{t:176.2}
Let $A$ and $B$ be infinite sets.
\begin{enumerate}
\item\label{t:176.2:part-rec} If $B \in \rec$, then $A$ is \frc{} to
$B$ if and only if $A$ is \frc{} to $\sigmastar$.
\item\label{t:176.2:part-re} If $B \in \re$, then $A$ is \fprc{} to
$B$ if and only if $A$ is \fprc{} to $\sigmastar$.
\end{enumerate}
\end{theorem}
\begin{proof}{}
We first prove part~\ref{t:176.2:part-rec}, beginning with the
``if'' direction.

Suppose $A$ is $\frec$-compressible to $\sigmastar$ by a recursive
function $f$, and suppose $B$ is recursive and infinite. Let $f'(x)$
output the element $y \in B$ such that $\rank_{B}(y) = f(x)$.  Then
$f'$ is recursive, and $A$ is \frc{} to $B$ by $f'$.

For the ``only if'' direction, let $B$ be an infinite recursive
set. Suppose that $A$ is \frc{} to $B$ by a recursive function
$f$. Let $f'(x)=\emptystring$ if $f(x)$ is not in $B$. Otherwise,
let $f'(x)=\rank_B(f(x))$. Then $f'$ is recursive, and $A$ is
\frc{} to $\sigmastar$ by $f'$.

Let us turn to part~\ref{t:176.2:part-re} of the theorem.  Again, we
begin with the ``if'' direction. Let $B$ be an infinite $\re$ set,
and let $E$ enumerate the elements of $B$ without
repetitions. Suppose $A$ is \fprc{} to $\sigmastar$ by a partial
recursive function $f$. Then $f'$ does the following on input $x$.
\begin{enumerate}
\item Simulate $f(x)$. This may run forever if
$x \not \in \domain(f)$.
\item If $f(x)$ outputs a value, simulate $E$ until it enumerates
$f(x)$ strings.
\item Output the $f(x)$-th string enumerated by $E$.
\end{enumerate}

The function $f'$ is partial recursive, and $A$ is \fprc{} to $B$
via $f'$.

For the ``only if'' direction, let $B$ be infinite and $\re$ and let
$E$ be an enumerator for $B$.  Suppose $A$ is \fprc{} to $B$ via a
partial recursive function $f$. On input $x$, our $f'$ will work as
follows.
\begin{enumerate}
\item Simulate $f(x)$.
\item If $f(x)$ outputs a value, run $E$ until it enumerates
$f(x)$. This step may run forever if $f(x) \not \in B$.
\item Suppose $f(x)$ is the $l$th string output by $E$. Then output
the $l$th string in $\sigmastar$.
\end{enumerate}

$f'$ is partial recursive, and $A$ is \fprc{} to $\sigmastar$ by
$f'$.
\end{proof}

Theorem~\ref{t:176.2} covers the two most natural pairings of set
classes with function classes: recursive sets $B$ with $\frec$
compression, and $\re$\ sets $B$ with $\fre$ compression. What about
pairing recursive sets under $\fre$ compression, or $\re$ sets under
recursive compression? We note as the following theorem that one and a
half of the analogous statements hold, but the remaining direction
fails.

\begin{theorem}\label{c-now-t:180.2-and-180.3}

\begin{enumerate}
\item Let $A$ and $B$ be infinite sets and suppose that
$B \in \rec$. Then $A$ is \fprc{} to $B$ if and only if $A$ is
\fprc{} to $\sigmastar$.
\item Let $A$ and $B$ be infinite sets with $B \in \re$. If $A$ is
\fprc{} to $\sigmastar$, then $A$ is \fprc{} to $B$. In fact, we
may even require that the compression function for $A$ to $B$
satisfies $f(\sigmastar) = B$.

\item There are infinite sets $A$ and $B$ with $B \in \re$ such that
$A$ is \frc{} to $B$ but $A$ is not \frc{} to $\sigmastar$.
\end{enumerate}

\end{theorem}

\begin{proof}
The first part follows immediately from Theorem~\ref{t:176.2},
part~\ref{t:176.2:part-re}.  The second part follows as a corollary
to the proof of Theorem~\ref{t:176.2}, part~\ref{t:176.2:part-re}.
In particular, the proof of the ``$\Leftarrow$'' direction proves
the second part, since it is clear that if $f$ is a recursive
function the $f'$ defined there is also recursive.

The third part follows from~\cite{hem-rub:jtoappear-with-tr-pointer:rft-compression1} in
which it is shown that any set in $\re - \rec$ is not \frc{} to
$\sigmastar$. Thus if we let $A = B$ be any set in $\re-\rec$, then
A is \frc{} to $B$ by the function $f(x) = x$ but $B$ is not \frc{}
to $\sigmastar$.
\end{proof}

Another interesting question is how recursive compressibility to $B$
is, or is not, linked to recursive isomorphism.
Recall two sets $A$ and $B$ are recursively isomorphic if there exists
a recursive bijection $f:\sigmastar \to \sigmastar$ with $f(A) = B$.
Although recursive
isomorphism of sets implies mutual compressibility to each other, we
prove via a finite-injury priority argument that the converse does not
hold (even when restricted to the $\re$ sets). The argument has an
interesting graph-theoretic flavor, and involves 
queuing infinitely
many strings to be added to a set at once.

\begin{theorem}\label{t:183.2-dan-3-28}
If $A \equiv_{\textit{iso}} B$, then $A$ is \frec{}-compressible to $B$ and $B$ is
\frec{}-compressible to $A$.
\end{theorem}

\begin{theorem}\label{t:priority}
There exist $\re$ sets $A$ and $B$ such that $A$ is
$\frec$-compressible to $B$ and $B$ is $\frec$-compressible to $A$,
yet $A \not\equiv_{\textit{iso}} B$.
\end{theorem}

\begin{proof}[Proof of Theorem~\ref{t:183.2-dan-3-28}]
Now $A$ is \frc{} to $B$ by simply letting our $\frec$-compression
function be the recursive isomorphism function $f$.  Since each
recursive isomorphism has a recursive inverse, $B$ is \frc{} to $A$
by letting our $\frec$-compression function be the inverse of $f$.
\end{proof}

\noindent{%
\color{darkgray}\sffamily\bfseries{}Proof of
Theorem~\ref{t:priority}.}\ 
Before defining $A$ and $B$, we will define a function $f$ which will
serve as both a compression function from $A$ to $B$ and a compression
function from $B$ to $A$.
First, fix a recursive isomorphism between $\sigmastar$ and
$\{\langle t,j,k \rangle \condition t \in \{0,1,2,3\} \land j,k \in
\N\}$.  Now we will define $f$ as follows. For each $j,k \in \N$, let
$f(\langle 3,j,k\rangle) = \langle 3,j+1,k\rangle$.  For each
$j,k \in \N$, $j>0$, and $t \in \{0,1,2\}$, let
$f(\langle t,j,k\rangle) = \langle t,j-1,k\rangle$.  Finally, for each
$k \in \N$, let $f(\langle 0,0,k\rangle) = \langle 3,0,k\rangle$,
$f(\langle1,0,k\rangle) = \langle0,0,k\rangle$, and
$f(\langle2,0,k\rangle) = \langle3,0,k\rangle$.  Let
$\ell:\sigmastar \to \{0,1\}$ be the unique function such that
$\ell(\langle0,0,k\rangle) = 0$ for all $k \in \N$ and
$\ell(f(x)) = 1-\ell(x)$. Let $D_f$ be the directed graph with edges
$(x,f(x))$. Note that $\ell$ is a 2-coloring of $D_f$ if we treat the
edges as being undirected. See Figure \ref{f:df}.

Call a set $C$ a \emph{path set} if for all $x \in C$, $f(x) \in C$
and there is exactly one $y \in C$ such that $f(y) = x$.  Suppose $C$
is a path set. Let $C_i = \{x \in C \condition \ell(x) = i\}$ for
$i \in \{0,1\}$. By the assumed property of $C$, we have $C_0$ and
$C_1$ are \frc{} to each other by $f$.  Furthermore, if $C$ is $\re$
then so are $C_0$ and $C_1$ since
$C_i = C \cap \{x \condition \ell(x) = i\}$ is the intersection of an
$\re$ set with a recursive set.  If we provide an enumerator for a
path set $C$ such that $C_0 \not \equiv_{\textit{iso}} C_1$, we may
let $A = C_0$ and $B = C_1$ and be done.

Our enumerator for $C$ proceeds in two interleaved types of stages:
printing stages $P_i$ and evaluation stages $E_i$. More formally, we
proceed in stages labeled $E_i$ and $P_i$ for $i \geq 1$, interleaved
as $E_1$,$P_1$,$E_2$,$P_2$,$\dots$,$E_n$,$P_n$,$\dots$ when
running. We also maintain a set $Q$ of elements of the form
$\langle t,k \rangle$, where $t \in \{0,1,2\}$ and $k \in \N$. This
set $Q$ will only ever be added to as the procedure runs.
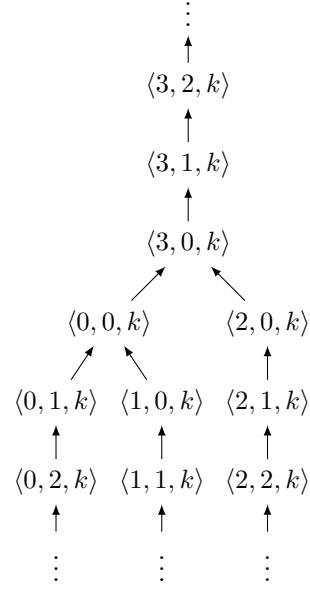
\begin{wrapfigure}{r}{0pt}
\tikzset{edge from parent/.style={draw,latex-}}
\begin{tikzpicture}[ level 4/.style={sibling distance=6em}, level
5/.style={sibling distance=4em}, level distance=3em ]
\node{$\vdots$} child { node{$\langle 3,2,k \rangle$} child {
node{$\langle 3,1,k \rangle$} child {
node{$\langle 3,0,k \rangle$} child {
node{$\langle 0,0,k \rangle$} child {
node{$\langle 0,1,k \rangle$} child {
node{$\langle 0,2,k \rangle$} child {
node{$\vdots$}}}} child {
node{$\langle 1,0,k \rangle$} child {
node{$\langle 1,1,k \rangle$} child {
node{$\vdots$}}}} } child {
node{$\langle 2,0,k \rangle$} child {
node{$\langle 2,1,k \rangle$} child {
node{$\langle 2,2,k \rangle$} child {
node{$\vdots$}}}}} }}} ;
\end{tikzpicture}
\caption{A diagram of $D_f$, for fixed $k$.}
\label{f:df}
\end{wrapfigure}

In the printing stage $P_i$, we do the following for every
$\langle t,k\rangle$ in $Q$. Enumerate $\langle3,j,k\rangle$ and
$\langle t,j,k\rangle$ for all $j \leq i$. If $t=1$, additionally
enumerate $\langle 0,0,k\rangle$.
Adding an element $\langle t,k\rangle$ to $Q$ in some evaluation stage
$E_i$ is essentially adding an infinite path of nodes in $D_f$ to $C$.

In addition to $Q$, we also maintain an integer $b$ and a set $R$ of
elements $\langle n, k \rangle$ where $n,k \in \N$. If
$\langle n,k\rangle \in R$ after stage $i$, it signifies that we have
not yet satisfied the condition that $\varphi_n$, the $n$th partial
recursive function, is not an isomorphism function between $C_0$ and
$C_1$. In stage $E_i$ we perform the following.  Add
$\langle i,b\rangle$ to $R$.  Increment $b$ by one.  For each
$\langle n,k\rangle \in R$, run $\varphi_n$, the $n$th partial
recursive function, on $\langle 0,0,k\rangle$ for $i$ steps.  If none
of these machines halt in their allotted time, end the
stage. Otherwise, let $n_i$ be the smallest number such that
$\varphi_{n_i}$ produced an output $w_i = \langle x_i,y_i,z_i\rangle$
on its respective input $\langle 0,0,k_i\rangle$. We now break into
cases:
\begin{enumerate}
\item If $\ell(w_i) = 0$ add $\langle 0,k_i \rangle$ to $Q$.

\item If $z_i \neq k_i$ and $\ell(w_i) = 1$ and as it stands $w_i$
would not be printed eventually if there were only type $P$ stages
from now on, add $\langle 0,k_i,\rangle$ to $Q$.
\item If $z_i \neq k_i$ and $\ell(w_i) = 1$ and as it stands $w_i$
would be printed eventually if there were only type $P$ stages from
now on, do nothing.
\item If $z_i = k_i$ and $\ell(w_i) = 1$ and $x_i = 0$, add
$\langle 1,k_i\rangle$ to $Q$.
\item If $z_i = k_i$ and $\ell(w_i) = 1$ and either $x_i = 1$ or
$x_i = 2$, add $\langle 0,k_i\rangle$ to $Q$.
\item If $z_i = k_i$ and $\ell(w_i) = 1$ and $x_i = 3$, add
$\langle 2,k_i\rangle$ to $Q$.
\end{enumerate}
Set $b = \max(k_i,z_i)+1$. Remove all pairs $\langle n,k \rangle$ with
$n \geq n_i$ from $R$. Then for each $n$ from $n_i+1$ to $i$, first
add $\langle n,b \rangle$ and subsequently increment $b$ by 1.

We will first prove that $C$ is a path set. If $x \in C$, then it is
printed in some printing stage $P_i$. By tracing the definition of $f$
and the procedure for printing stages, one can verify that both
$f(x)$ and exactly one $y$ such that $f(y) = x$ will be printed in
stage $P_{j}$ for $j \geq i$. This string $y$ will be the only one
ever printed, since no two elements with the same second coordinate
will ever be added to $Q$, as every element added to $Q$ has the
current state of $b$ as its second coordinate, and $b$ only ever
strictly increases between additions to $Q$.

Let $F_n$ be the condition that $\varphi_n$ fails to be a recursive
isomorphism of $C_0$ onto $C_1$. Fix $n$. Say during $E_i$ we have
$n_i = n$.  In cases 1, 2, 4, and 5, we force $\varphi_{n}$ to map
$\langle 0,0,k_i \rangle \in C_0$ to something out of $C_1$.  In cases
3 and 6, we force $\varphi_{n}$ to map
$\langle 0,0,k_i\rangle \notin C_0$ to something in $C_1$. Thus
whenever at stage $i$ we have $n_i = n$, condition $F_n$ becomes
satisfied, though perhaps not permanently. Specifically, in case 2,
$w$ could be printed later to satisfy some other $F_{m}$ and in doing
so ``injure'' $F_n$. However, note that during $E_i$ the variable $b$ is set
to $\max(k_i,z_i)$, thus $F_n$ can only be injured when satisfying
conditions $F_m$ for $m < n$.  Pairs with first coordinate $n$ will
only ever be added to $R$ when after satisfying some such $F_m$, in
addition to once initially, so in total only a finite number of times.
If $\varphi_{n}$ always halts, $F_n$ will eventually be satisfied and
never injured again.

This proves that $C$ is a path set such that
$C_0 \not\equiv_{\mathit{iso}} C_1$. Thus $C_0$ and $C_1$ are $\re$
sets that are \frc{} to each other by $f$, but are not recursively
isomorphic.
\ \hfill{\textcolor{darkgray}{\ensuremath{\blacktriangleleft}}}

For those interested in the issue of isomorphism in the context of
complexity-theoretic functions, which was not the focus above, we
mention that: Hemaspaandra, Zaki, and
Zimand~\cite{hem-zak-zim:c:semi-rankable} prove that the P-rankable
sets are not closed under~$\equiv^p_{\textit{iso}}$; Goldsmith and
Homer~\cite{gol-hom:j:scalability} prove that the strong-P-rankable
sets are closed under $\equiv^p_{\textit{iso}}$ if and only if
$\p=\psharpp$; and \cite{hem-zak-zim:c:semi-rankable} notes that
the semistrong-P-rankable sets similarly
are closed under $\equiv^p_{\textit{iso}}$ if and only if
$\p=\psharpp$.

\section{Closures and Nonclosures under Boolean
Operations}\label{s:boolean}

We now move on to a main focus of this paper, the closure properties
of the compressible and the rankable sets. We explore these properties
both in the complexity-theoretic and the recursion-theoretic
domains. Table \ref{table} on page~\pageref{table} summarizes our
findings.

\begin{lemma}\label{lem10} Let $A$ and $B$ be \spr{}. Then $A \cup B$
is \spr{} if and only if $A \cap B$ is.
\end{lemma}
\begin{proof}
The identity
$\rank_{A \cap B}(x) + \rank_{A \cup B} = \rank_A(x) + \rank_B(x)$
allows us to compute either of $\rank_{A \cap B}(x)$ or
$\rank_{A \cup B}(x)$ from the other.
\end{proof}

\begin{theorem}\label{thm11} The following conditions are equivalent:
\begin{enumerate}
\item the classes \spr{} and \pr{} are closed under intersection,
\item the classes \spr{} and \pr{} are closed under union, and
\item $\p = \psharpp$.
\end{enumerate}
\end{theorem}
\begin{proof} It was proven in~\cite{hem-rud:j:ranking} by Hemachandra
and Rudich that $\p = \psharpp$ implies
$\p= \text{\spr{}}=\allowbreak\text{\pr{}}$. Since $\p$ is closed
under intersection and union, this shows that 3 implies 1 and 2.
To show, in light of Lemma \ref{lem10}, that either 1 or 2 would
imply 3, we will construct two \spr{} sets whose intersection is not
\wpr{} unless $\p = \psharpp$.

Let $A_1$ be the set of $x1y1$ such that $|x| = |y|$, $x$ encodes a
boolean formula, and $y$ (padded with 0s so that it has length
$|x|$) encodes a satisfying assignment for the formula $x$. Let
$A_0$ be the set of $x1y0$ such that $|x| = |y|$, and
$x1y1\notin A_1$.
Let $A_2$ be the set of strings $x0^{|x|+1}1$. Let
$A = A_0 \cup A_1 \cup A_2$.  For every $x$, and every $y$ such that
$|x|=|y|$, exactly one of $x1y0$ and $x1y1$ is in $A$.
Thus, for any $x$, we can find $\rank_{A_0\cup A_1}(x)$ in
polynomial time. Clearly $A_2$ is \spr{}. Since $A_0\cup A_1$ and
$A_2$ are disjoint,
$\rank_{A_0\cup A_1 \cup A_2}(x)=\rank_{A_0\cup
A_1}(x)+\rank_{A_2}(x)$, so $A$ is \spr{}.

Let $B = \sigmastar1$.
Then $A \cap B = A_1 \cup A_2$ is the set of $x1y1$ such that $y$
encodes a satisfying assignment for $x$, along with all strings
$x0^{|x|+1}1$. If $A_1 \cup A_2$ were \wpr{}, then we could count
satisfying assignments of a formula $x$ in polynomial time by
computing
$\rank_{A \cap B}(\shift(x,1)0^{|\shift(x,1)|+1}1) - \rank_{A \cap
B}(x0^{|x|+1}1)-1$. Thus $\operatorname{\#SAT}$ is polynomial-time
computable and so $\p = \psharpp$.
\end{proof}

\begin{proposition}\label{prop-was-lem-lem13} \spr{} is closed under
complementation.
\end{proposition}
\begin{proof} The identity
$\rank_{A}(x) + \rank_{\overline{A}}(x) =\rank_{\sigmastar}(x)$
allows us to compute either of $\rank_{A}(x)$ or
$\rank_{\overline{A}}(x)$ from the other.
\end{proof}
\begin{corollary}\label{cor14} The class \nsprit{} is also closed under
complementation.
\end{corollary}
\begin{lemma}\label{lem15} The class \pr{} is closed under complementation if
and only if $\allowbreak$\pr{} = \spr{}.
\end{lemma}
\begin{proof}
The ``if'' direction follows directly from Proposition
\ref{prop-was-lem-lem13}. For the ``only if'' direction, let $A$ be
a \pr{} set with ranking function $r_A$, and suppose $\overline{A}$
is \pr{} with semistrong ranking function $r_{\overline{A}}$. Then
	$\rank_A(x) = r_A(x)$ if $x \in A$, and equals $\rank_{\sigmastar}(x)-r_{\overline{A}}(x)$
otherwise. The function $r_A$ decides membership in $A$, so we can
compute $\rank_A(x)$ in polynomial time.
\end{proof}

\begin{theorem}\label{thm16} If \pr{} is closed under complementation,
then $\p = \up\cap\coup$.
\end{theorem}
\begin{proof}
Suppose \pr{} is closed under complementation. Let $A$ be in
$\up\cap\coup$. Then there exists a UP machine $U$ recognizing $A$,
and a UP machine $\hat{U}$ recognizing $\overline{A}$. If $x \in A$,
let $f(x)$ be the unique accepting path for $x$ in $U$. Otherwise,
let $f(x)$ be the unique accepting path for $x$ in $\hat{U}$. Choose
a polynomial $p$ such that, without loss of generality, $p(x)$ is
monotonically increasing and $|f(x)| = p(|x|)$ (we may pad accepting
paths with 0s to make this true).

The language
$B = \{xf(x)1 \condition x \in \sigmastar\} \cup \{x0^{p(|x|)+1}
\condition x \in \sigmastar\}$ is \pr{} since
$\rank_{B}(x0^{p(|x|)+1}) = 2 \rank_{\sigmastar}(x)-1$ and
$\rank_B(xf(x)1) = 2 \rank_{\sigmastar}(x)$.  Since \pr{} is closed
under complementation, and $B$ is \pr{}, $B$ is also \spr{} by
Lemma~\ref{lem15}. Let $x$ be a string, and let $y =
\shift(x,1)$. We can binary search on the value of $\rank_B$ in the
range from $x0^{p(|x|)+1}$ to $y0^{p(|y|)+1}$ to find the first
value $xz$ where $|z| = p(|x|)+1$ and
$\rank_{B}(xz) = 2\rank_{\sigmastar}(x)$.  See that $f(x)$ must
equal $z$. We then simulate $U$ on the path $z$ and $\hat{U}$ on the
path $z$. Now $z$ must be an accepting path for one of these
machines, so either $U$ accepts and $x \in A$, or $\hat{U}$ accepts
and $x \notin A$.
\end{proof}

\begin{definition}\label{def17}
A set is \emph{nongappy} if there exists a polynomial $p$ such that,
for each $n \in \N$, there is some element $y \in A$ such that
$n \leq |y| \leq p(n)$.
\end{definition}

\begin{theorem}\label{changed-thm18}
If $\p = \up\cap\coup$ then
each nongappy \pr{} set is \spr{}.
\end{theorem}
\begin{proof}
Let $A$ be a nongappy \pr{} set, and let $p$ be a polynomial such
that, for each $n \in \N$, there is $y$ in $A$ such that
$n \leq |y| \leq p(y)$.  Let $r$ be a polynomial-time semistrong
ranking function for $A$.  The coming string comparisons of course
will be lexicographical.  Let $L$ be the set of $\langle x,b\rangle$
such that there exists at least one string in $A$ that is
less than or equal to $x$ and $b$ a prefix of the greatest string in
$A$ that is lexicographically less than or equal to $x$.  $L$ is in
$\up \cap \coup$ by the following procedure.  Let $x_0$ be the
lexicographically first string in $A$.  If
$x < x_0$ output 0.  Otherwise, guess a string $z > x$ such that
$|z| \leq p(|x|+1)$. Then guess a $y \leq x$. If $y$ and $z$ are in
$A$ and $r(y)+1 = r(z)$,
then we know that
and $y$ and $z$ are the (unique) strings in $A$ that most tightly
bracket $x$ in the $\leq$ and the $>$ directions.
We can in our current case build the greatest string less than or
equal to $x$ that is in $A$ bit by bit, querying potential prefixes,
in polynomial time. Since $\rank_A(x) = \rank_A(y)$, we can compute
$\rank_A(x)$ in polynomial time for arbitrary $x$.
\end{proof}

{}From Proposition~\ref{prop-was-lem-lem13} and
Theorem~\ref{changed-thm18}, we obtain the following corollary.
\begin{corollary}\label{c:new-cite-for-old-thm18}
If $\p = \up\cap\coup$ then the complement of each nongappy
\pr{} set is \spr{} (and so certainly is \pr).
\end{corollary}

\begin{theorem}\label{yellow}
There exist \wpr{} sets $A$ and $B$ such that $A \cap B$ is infinite
but not $\fre$-compressible.
\end{theorem}
\begin{proof}
We will define a set $A$ not containing the empty string and
satisfying the condition that for all $x \in \sigmastar$, exactly
one of $x0$ and $x1$ is in $A$. Then clearly $A$ is \wpr{} by a
compression function sending $x1$ and $x0$ to
$\rank_{\sigmastar}(x)$. Let $A_0$ and $B_0$ be empty, and let
$m_0 = \emptystring$. We will define $A_i$, $B_i$, and $m_i$
inductively for $i > 0$. Let $\varphi_i$ be the $i$th Turing machine
in some enumeration of all Turing machines.
\begin{enumerate}
\item Suppose that $\varphi_i$ is defined on $m_{i-1}0$, and that
for all
$x \in \left(A_{i-1} \cap B_{i-1}\right) \cup \{y \condition y >
m_{i-1}0\}$ we have $\varphi_i(x) \neq \varphi_i(m_{i-1}0)$.
In this case, we set
$A_{i} = A_{i-1} \cup \{m_{i-1}0,\shift(m_{i-1},1)0\}$ and
$B_{i} = B_{i-1} \cup \{m_{i-1}1,\shift(m_{i-1},1)0\}$ and set
$m_i = \shift(m_{i-1},2)$, so that neither $m_{i-1}0$ nor
$m_{i-1}1$ is in $A_{i} \cap B_{i}$. Note that
$\shift(m_{i-1},1)0 \in A_i \cap B_i$ but
$\shift(m_{i-1},1)0 \notin A_{i-1} \cap B_{i-1}$.
\item Suppose $\varphi_i$ is either undefined on $m_{i-1}0$, or that
for some $x \in A_{i-1} \cap B_{i-1}$ we have
$\varphi_i(x) = \varphi_i(m_{i-1}0)$. In this case, set
$A_i = A_{i-1} \cup \{m_{i-1}0\}$,
$B_i = B_{i-1} \cup \{m_{i-1}0\}$, and
$m_i=\shift(m_{i-1},1)$. Note in particular that $x$ and
$m_{i-1}0$ are both in $A_i \cap B_i$ and lexicographically less
than $m_i0$, and take the same value under $\varphi_i$.
\item Suppose that the above cases do not hold and there is some
$x > m_{i-1}1$ such that $\varphi_i(x) = \varphi_i(m_{i-1}0)$. Let
$y$ be the lexicographically largest string such that $y0 \leq x$,
and let $m_i = \shift(y,1)$. Set
$A_i = A_{i-1} \cup \{z0\condition m_{i-1} \leq z < y\} \cup
\{x\}$ and
$B_i = B_{i-1} \cup \{z0 \condition m_{i-1} \leq z < y\} \cup
\{x\}$. Note in particular that $x$ and $m_{i-1}0$ are both in
$A_i \cap B_i$ and lexicographically less than $m_i0$, and take
the same value under $\varphi_i$.
\end{enumerate}

Finally, let $A = \bigcup_{i \geq 0} A_i$ and
$B = \bigcup_{i \geq 0} B_i$. Notice that stage $i$ only adds
elements to $A_i$ or $B_i$ that are lexicographically greater than
or equal to $m_{i-1}0$, so if $x < m_i0$ and
$x \notin A_i \cap B_i$, then $x \notin A \cap B$. In case 1, we see
that $\varphi_i$ fails to be surjective (i.e, onto $\sigmastar$)
when restricted to $A \cap B$, since there is no $x < m_i0$ in
$A \cap B$ mapping to $\varphi_i(m_{i-1}0)$, and also no
$x > m_{i-1}1$ mapping to $\varphi_i(m_{i-1}0)$, and neither
$m_{i-1}0$ nor $m_{i-1}1$ is in $A \cap B$.  In case 2, we see
either that $\varphi_i$ is undefined on an element of $A \cap B$ or
that two elements of $A \cap B$ map to the same element. In case 3,
we see that two elements in $A \cap B$ map to the same element under
$\varphi_i$. Thus $\varphi_i$ fails to compress $A \cap B$, and no
partial recursive function can compress $A \cap B$. The set
$A \cap B$ is infinite since at least one new element is added to
$A_i \cap B_i$ during stage $i$. We also maintain the condition
that, for all $x < m_i$, exactly one of $x0$ and $x1$ is in $A_i$
(resp.,\ $B_i$). Each $A_i$ (resp. $B_i$) consists of exactly all
strings in $A$ (resp. $B$) lexicographically less than $m_i0$, and
so clearly since this statement holds for each $A_i$ (resp. $B_i$)
it holds for all of $A$ (resp. $B$) as well.  Thus $A$ and $B$
are \wpr{}, but their intersection is not $\fre$\tcomp.
\end{proof}
\begin{theorem}\label{red}
There exist infinite \wpr{} sets $A$ and $B$ such that $A \cup B$ is
not $\fre$-compressible.
\end{theorem}

\begin{theorem}\label{blue}
There exists an infinite \wpr{} set whose complement is infinite but
not $\fre$-compressible.
\end{theorem}

\begin{theorem}\label{green}
There exist sets $A$ and $B$ that are not \fprc{}, yet
$A \cup B$ is \spr{}. In addition, there exist sets $A$ and $B$ that
are not \fprc{},
yet $A \cap B$ is \spr{}.
\end{theorem}
The proofs of these three theorems are
in the appendix.

\section{Additional Closure and Nonclosure Properties}\label{s:join}

How robust are the polynomial-time and recursion-theoretically
compressible and the rankable sets?  Do sets lose these properties
under join, or subtraction, addition, or (better yet) symmetric
difference with finite sets?  Or even with sufficiently nice infinite
sets?  The following section addresses these questions.

\subsection{Complexity-Theoretic Results}

We focus on the join (aka disjoint union), giving a full
classification of the closure properties (or lack thereof) of the
\wpr{}, \pr{}, and \spr{} sets, as well as their complements, under
this operation.  The literature is not consistent as to whether the
low-order or high-order bit is the ``marking'' bit for the join.
Here, we
follow the classic computability texts of Rogers~\cite{rog:b:rft} and
Soare~\cite{soa:b:degrees} and the classic structural-complexity text
of Balc{\'{a}}zar, D{\'{\i}}az,
Gabarr{\'{o}}~\cite{bal-dia-gab:b:sctI-2nd-ed}, and define the join
using low-order-bit marking: The join of $A$ and $B$, denoted
$A \join B$, is $A0 \cup B1$, i.e.,
$\{x0 \condition x \in A\} \cup \{x1 \condition x \in B\}$.
For classes invariant under
reversal,
which end is used for the marking bit is not important (in the sense
that the class itself is closed under upper-bit-marked join if and
only if it is closed under lower-bit-marked join).  However,
the placement of the marking bit potentially matters for ranking-based
classes, since those classes are based on lexicographical order.

The join is such a basic operation that it seems very surprising that
any class would not be closed under it, and it would be even more
surprising if the join of two sets that lack some nice organizational
property (such as being \wpr) can have that property (can be \wpr, and
we indeed show in this section that that happens)---i.e., the join of
two sets can be ``simpler'' than either of them (despite the fact that
the join of two sets is the least upper bound for them with respect to
$\leq_m^p$~\cite{sch:b:complexity-structure}, and in the sense of
reductions captures the power-as-a-target of both sets).  However,
there is a precedent for this in the literature,
and it regards a rather important complexity-theoretic structure.
It
is known that $(\rm EL_2)^\complement$ is not closed under the
join~\cite{hem-jia-rot-wat:j:join-lowers}, where $\rm EL_2$ is the
second level of the extended low hierarchy~\cite{bal-boo-sch:j:low}.

\begin{theorem}\label{thm30}
If $\p\neq\psharpp$ then there exist sets $A\in\p$ and $B\in\p$ that
are not \wpr{} yet $A\cap B$, $A\cup B$, and $A \join B$ are \spr{}.
\end{theorem}
\begin{proof}
In this proof we construct a set $A_1$ whose members represent
satisfying assignments of boolean formulas. When we force certain
elements, or beacons, into $A_1$ we obtain a set $A$ such that if we
were able to rank $A$, we could count the number of satisfying
assignments to a boolean formula by comparing the rank of these
beacons. The set $B$ is constructed similarly, but in a way that
$A \cup B$, $A \cap B$, and $A \join B$ are easily \spr{}.

As in the rest of the paper, $\Sigma=\{0,1\}$. Let
$A_1=\{\alpha 01 \beta\condition \alpha, \beta\in\Sigma^* \land
|\alpha|=|\beta| \land \alpha$ is a valid encoding of boolean formula $F$
that has (without loss of generality) $k\leq|\alpha|$ variables, the
first $k$ bits of $\beta$ encode a satisfying assignment of $F$, and
the rest of the $|\beta|-k$ bits of $\beta$ are 0$ \}$.
Note that given a string $x = \alpha 01 \beta \in A_1$, we can unambiguously extract $\alpha$ and $\beta$ because they must have length $(|x|-2)/2$.
Let
$B_1=\{\alpha01\beta\condition \alpha,\beta\in\Sigma^* \land
|\alpha|=|\beta| \land \alpha01\beta\notin A_1\}$.  Let
$\mathit{Beacons}=\{\alpha000^{|\alpha|} \condition
\alpha\in\Sigma^*\}\cup\{\alpha110^{|\alpha|}
\condition\alpha\in\Sigma^*\}$.
Similarly to $A_1$, strings in $B_1$ and $\mathit{Beacons}$ can be parsed unambiguously.
Let $A=A_1\cup
\mathit{Beacons}$. Let $B=B_1\cup \mathit{Beacons}$.  Note that $A$
and $B$ are both in $\p$ because checking if an assignment satisfies
a boolean formula is in $\p$ and $\mathit{Beacons}$ is clearly in
$\p$.

We will now demonstrate that if either $A$ or $B$ were \wpr{}, then
$\operatorname{\#SAT}$ would be in $\p$. Suppose that $A$ is \wpr{}
and let $f$ be a polynomial-time ranking function for $A$. Let
$\alpha$ be a string encoding a boolean formula $F$. Then we can
compute $j=f(\alpha110^{|\alpha|})-f(\alpha000^{|\alpha|})$
in polynomial time. Both $\alpha110^{|\alpha|}$ and
$\alpha000^{|\alpha|}$ are in $\mathit{Beacons}$ and thus in $A$,
so $f$ gives a true ranking for these values. Every string in $A$
between (and not including) these $\mathit{Beacons}$ strings is from
$A_1$ and thus represents a satisfying assignment for $F$, and every
satisfying assignment for $F$ is represented by a string between
these $\mathit{Beacons}$ strings. Because the last $|\beta|-k$ bits
of $\beta$ are 0, where $k$ is the number of variables in $F$, each
satisfying assignment for $F$ is represented exactly once between the
two $\mathit{Beacons}$ strings. Thus $j-1$ is the number of
satisfying assignments of $F$. We can compute $j$ in polynomial
time, so $\operatorname{\#SAT}$
is polynomial-time computable and thus $\p=\psharpp$, contrary to
our $\p\neq\psharpp$ hypothesis.

Now suppose that $B$ is \wpr{} and similarly to before we will let
$f$ be the $\p$-time ranking function for it. Again we will let
$\alpha$ be the encoding for some boolean formula $F$ and
$j=f(\alpha110^{|\alpha|})-f(\alpha000^{|\alpha|})$. In this
case the strings in $B$ between $\alpha110^{|\alpha|}$ and
$\alpha000^{|\alpha|}$ are the strings of the form
$\alpha01\Sigma^{|\alpha|}$ except for those that are in $A_1$
(and recall that those that are in $A_1$ are precisely the
padded-with-0s satisfying assignments for $F$).
Because we know the number of strings of the form
$\alpha01\Sigma^{|\alpha|}$, we can again find the number of
satisfying assignments for $F$. Namely, we have that
$j=1+2^{|\alpha|}-s$, where $s$ is the number of satisfying
assignments of $F$. Thus if $B$ is \wpr{}, then we can find $s$ in
polynomial time and thus $\p=\psharpp$, contrary to our
$\p\neq\psharpp$ hypothesis.

Finally, we show that $A\cup B$, $A \cap B$, and $A \join B$ are
\spr{}.
The set $A\cap B$ is simply $\mathit{Beacons}$, which is
\spr{} as follows. 
Any string lexicographically below $00$ has rank 0.
For any $\alpha\in\Sigma^*$, the rank of
$\alpha000^{|\alpha|}$ is $2\rank_{\sigmastar}(\alpha)-1$ and
the rank of $\alpha110^{|\alpha|}$ is
$2\rank_{\sigmastar}(\alpha)$.
For every other string, it is easy to
find the lexicographically greatest
string in $\mathit{Beacons}$ that is lexicographically less than
the given string in polynomial time, and so it is possible to rank the string in polynomial time.

The set
$A\cup B= \{\alpha01\beta\condition \alpha\in\Sigma^* \land
\beta\in\Sigma^*\land|\alpha|=|\beta|\}
\cup\{\alpha000^{|\alpha|}\condition \alpha\in\Sigma^*\}
\cup\{\alpha110^{|\alpha|}\condition \alpha\in\Sigma^*\}$,
and is also \spr{}, as follows.
Any string lexicographically below $00$ has
rank 0. For any $\alpha\in\Sigma^*$, the rank of
$\alpha000^{|\alpha|}$ is
$1+\sum_{x<_{lex}\alpha} (2^{|x|}+2)$, where $x<_{lex}\alpha$ denotes
that $x$ is lexicographically less than $\alpha$. Note that although
the sum is over an exponentially sized set, it still can be computed
in polynomial time because the summands depend only on the length of
the element in the set. Let $b(x)$ be the number of strings lexicographically less than $\alpha$ but with the same length as $\alpha$.
Then we have that
$1+\sum_{x<_{lex}\alpha}(2^{|x|}+2)=1+b(\alpha)(2^{|\alpha|}+2) + \sum_{i=0}^{|\alpha|-1}(2^i(2^i+2))$.

The rank of $\alpha110^{|\alpha|}$ is
$\sum_{x\leq_{lex}\alpha}(2^{|x|}+2)$, where $x\leq_{lex}\alpha$ denotes that $x$ is lexicographically less than or equal to $\alpha$.  For any
$\alpha, \beta \in \Sigma^*$ where $|\alpha|=|\beta|$, the rank of
$\alpha01\beta$ is $b(\beta) + 2 + \sum_{x<_{lex}\alpha}(2^{|x|}+2)$,
where $n$ is the integer such that $\beta$ is the $n$th string of
its length. As above, each term is only dependent on the length of
$x$, and is computable in polynomial time. For any other not string in
$A\cup B$, it is easy to
find the greatest string in $A \cup B$ lexicographically less than
the given string in polynomial time, and thus it is easy to rank that string.

We can show that
$A \join B=\{a0\condition a\in A\}\cup\{b1\condition b\in B\}$ is
\spr{} using the fact that $A\cup B$ and $A \cap B$ are \spr{}, and
both $A$ and $B$ are in $\p$.
The rank of $\epsilon$ is 0.  The rank of 0 is
1 if $\epsilon\in A$ and otherwise is $0$.
For $x \in \sigmastar$, we have $\rank_{A \join B}(x1)= \rank_{A \cup B}(x)+\rank_{A \cap B}(x)$. For $x \neq \emptystring$, we have $\rank_{A \join B}(x0) = \rank_{A \join B}(x1) - \delta_B(x)$, where $\delta_B(x) = 1$ if and only if $x \in B$.
\end{proof}
\begin{theorem}
The following are equivalent:
\begin{enumerate}
\item \nsprit{} is closed under join,
\item \nprit{} is closed under join, and
\item $\p=\psharpp$.
\end{enumerate}
\end{theorem}
\begin{proof}
Theorem~\ref{thm30} shows that either of 1 or 2 would imply 3. Now
we show that 3 implies 1 and 2, or equivalently the negation of
either 1 or 2 would imply the negation of 3.  Suppose that \nspr{}
(resp.,\ \npr{}) is not closed under join.  Then there are two sets
$A$ and $B$ that are in \nspr{} (resp.,\ \npr{}) but $A\join B$ is
\spr{} (resp.,\ \npr{}).  Then both $A$ and $B$ are in $\p$.  This
is because $A\join B\in \p$ and to test $x$ for membership in $A$,
for example, we can just test $x0$ for membership in $ A\join B$. It
was shown by Hemachandra and Rudich~\cite{hem-rud:j:ranking} that
$\p=\psharpp$, $\p=\text{\pr}$, and $\p=\text{\spr{}}$ are
equivalent.
Since $A$ and $B$ are in $\p$ but not \spr{} (resp.,\ \pr{}),
$\p\neq\text{\spr{}}$ (respectively $\p\neq\text{\pr}$) and thus
$\p\neq\psharpp$.
\end{proof}
\begin{theorem}\label{t:cpo}
The class \nwprit{} is not closed under join.
\end{theorem}

\begin{theorem}\label{t:tyef}
The class \wpr{} is not closed under join.
\end{theorem}
\begin{theorem}\label{t:221}
The class \spr{} is closed under join.
\end{theorem}

\begin{theorem}\label{t:903}
The class \pr{} is closed under complement if and only if it is closed under
join.
\end{theorem}

\begin{theorem}\label{t:2243}
The class \pc{} is closed under join.
\end{theorem}

The proofs of Theorems
\ref{t:cpo}--\ref{t:2243}
are in the appendix.

\subsection{Recursion-Theoretic Results}
\begin{theorem}\label{t:rft:delta}
\begin{enumerate}
\item If $A$ is an \frr{} set,
$B_1 \subseteq A$ is a recursive set, and
$B_2 \subseteq \overline{A}$ is a recursive set, then
$A \bigtriangleup (B_1 \cup B_2)$ (equivalently,
$(A - B_1) \cup B_2$) is \frr{}.

\item   If $A$ is \frc{}, $B_1 \subseteq A$ is recursive, and $A-B_1$ contains an
infinite \re{} subset, then $A - B_1$ is \frc{}.

\item If $A$ is an \frc{} set and $B_2 \subseteq \overline{A}$ is a
recursive set, then $A \cup B_2$ is an \frc{} set.
\end{enumerate}
\end{theorem}

Theorem~\ref{t:rft:delta}'s proof is in the
appendix.

\begin{corollary}
\begin{enumerate}
\item The class of $\frec$-rankable sets is closed under symmetric
difference with finite sets (and thus also under removing and
adding finite sets).
\item The class of $\frec$-compressible sets is closed addition and
subtraction of finite sets.
\end{enumerate}
\end{corollary}

\section{Conclusions}
Taking to heart the work in earlier papers that views as classes the
collections of sets that have (or lack) rankability/compressibility
properties, we have studied whether
those classes
are closed under the most important boolean and other
operations.
For the studied classes,
we in almost every case were able to
prove that they are closed under the operation,
or
to prove that they are not closed under the operation,
or to prove that
whether they are closed depends on well-known questions
about the equality of standard complexity classes.
Additionally, we have
introduced the notion of compression onto a set and
have showed the
robustness of compression under this
notion, as well as the limits of that robustness.
Appendix~\ref{appendix:B:underdeveloped-parts} provides some
additional directions and some preliminary results on them.

\bibliography{gryMoreOn}

\appendix
\section{Appendix}\label{appendix}

In this section, we include proofs omitted from earlier sections.

\begin{proof} [Proof of Theorem~\ref{red}]Let a set $A$ not containing
the empty string satisfy the condition that for all
$x \in \sigmastar$, exactly one of $x0,~x1$ is in $A$. Then clearly
$A$ is \wpr{} by a function sending $x1$ and $x0$ to
$\rank_{\sigmastar}(x)$.

Let $A_0$ and $B_0$ be empty, and let $m_0 = \emptystring$. We will
construct $A_i$, $B_i$, and $m_i$ inductively for $i > 0$.  Let
$\varphi_i$ be the $i$th Turing machine in some enumeration of all
Turing machines.
\begin{enumerate}
\item Suppose that $\varphi_i$ is defined on $m_{i-1}0$, and that
for all $x$ in
$A_{i-1} \cup B_{i-1} \cup \{y \condition y > m_{i-1}0\}$ we have
$\varphi_i(x) \neq \varphi_i(m_{i-1}0)$.
In this case, we set $A_{i} = A_{i-1} \cup \{m_{i-1}1\}$ and
$B_{i} = B_{i-1} \cup \{m_{i-1}1\}$, and we set
$m_i = \shift(m_{i-1},1)$.
\item Suppose $\varphi_i$ is either undefined on $m_{i-1}0$, or that
for some $x \in A_{i-1} \cup B_{i-1}$ we have
$\varphi_i(x) = \varphi_i(m_{i-1}0)$. In this case, set
$A_i = A_{i-1} \cup \{m_{i-1}0\}$,
$B_i = B_{i-1} \cup \{m_{i-1}0\}$, and $m_i = \shift(m_i,1)$.
\item Suppose that the above cases do not hold and there is some
$x \geq m_{i-1}0$ such that $\varphi_i(x) =
\varphi_i(m_{i-1}0)$. Let $m_i$ be the lexicographically smallest
string such that $m_i0 > x$. Set
$A_i = A_{i-1} \cup \{y0 \condition m_{i-1} \leq y < m_i\}$ and
$B_i = B_{i-1} \cup \{y1 \condition m_{i-1} \leq y < m_i\}$. Note
that both $m_{i-1}0$ and $x$ are in $A_i \cup B_i$.
\end{enumerate}

Finally, let $A = \bigcup_{i \geq 0} A_i$ and
$B = \bigcup_{i \geq 0} B_i$. Stage $i$ only adds elements to $A_i$
or $B_i$ that are lexicographically greater than or equal to
$m_{i-1}0$, so if $x < m_i0$ and $x \notin A_i \cup B_i$, then
$x \notin A \cup B$.  In case 1, we see that $\varphi_i$ fails to be
surjective (i.e., onto $\sigmastar$) when restricted to $A \cup B$,
since there is no $x < m_{i-1}0$ in $A \cup B$ mapping to
$\varphi_i(m_{i-1}0)$, and also no $x > m_i0$ mapping to
$\varphi_i(m_{i-1}0)$, so no element in $A \cup B$ compresses to
$\varphi_i(m_{i-1}0)$.  In case 2, we see either that $\varphi_i$ is
undefined on $m_{i-1}0 \in A \cup B$ or that $m_{i-1}0$ and some
other element in $A \cup B$ map to the same element, so injectivity
when restricted to $A \cap B$ fails.  Similarly, in case 3, we see
that $m_{i-1}0$ and some other element in $A \cup B$ will map to the
same value under $\varphi_i$.  Thus for all $i$ we see that
$\varphi_i$ fails to compress $A \cup B$, and so no partial
recursive function compresses $A \cup B$. Note that we maintain the
condition that for all $x < m_i$, exactly one of $x0$ and $x1$ is in
$A_i$ (resp.,\ $B_i$). This condition holds in $A$ (resp.,\ $B$),
and this property carries over to $A$ and $B$ as well. Each $A_i$
(resp. $B_i$) consists of exactly all strings in $A$ (resp. $B$)
lexicographically less than $m_i0$, and so clearly since this
statement holds for each $A_i$ (resp. $B_i$) it holds for all of $A$
(resp. $B$) as well. Thus $A$ and $B$ are \wpr{}, but $A \cup B$ is
not \fprc{}.
\end{proof}

\begin{proof}[Proof of Theorem~\ref{blue}]
We will construct a set $A$ consisting of strings with length at
least 2, with the property that for every $x \in \sigmastar$,
exactly one of $x00$, $x01$, $x10$, and $x11$ is in $A$. Clearly $A$
will be infinite, and its complement is infinite as well. Also, $A$
will be \wpr{} by sending $x00$, $x01$, $x10$ and $x11$ to
$\rank_{\sigmastar}(x)$. Let $A_0 = 0$ and $m_0 = \emptystring$. We
will construct $A_i$ and $m_i$ inductively for $i > 0$.  Let
$\varphi_i$ be the $i$th Turing machine in some enumeration of all
Turing machines.
\begin{enumerate}
\item Suppose $\varphi_i$ halts on $m_{i-1}00$, and there is no
$x \in \overline{A_{i-1}}$ where $x < m_{i-1}00$ such that
$\varphi_i(x) = \varphi(m_{i-1}00)$, and that there is no
$x > m_{i-1}00$ such that $\varphi_i(x) =\varphi(
m_{i-1}00)$. Then set $A_i = A_{i-1} \cup \{m_i00\}$ and set
$m_i = \shift(m_{i-1},1)$.
\item Suppose $\varphi_i$ is undefined on $m_{i-1}00$, or that there
is some $x < m_{i-1}00$ where $x \in \overline{A_{i-1}}$ and
$\varphi(x) = \varphi(m_{i-1}00)$. Then set
$A_i = A_{i-1} \cup \{m_{i-1}01\}$ and set
$m_i = \shift(m_{i-1},1)$.
\item Suppose that the above cases do not hold, $\varphi_i$ is
defined on $m_{i-1}00$, and $\varphi_i(m_{i-1}00) = \varphi(x)$
for some $x \in \{m_{i-1}01$, $m_{i-1}10$, $m_{i-1}11\}$. Then set
$A_i = A_{i-1} \cup \{z\}$, where $z$ is a fixed arbitrary element
in $\{m_{i-1}01$, $m_{i-1}10$, $m_{i-1}11\}-\{x\}$, and set
$m_{i} = \shift(m_{i-1}$, $1)$. Note that $m_{i-1}00$ and $x$ are
both in $\overline{A}$ and below $m_i00$, and take the same value
under $\varphi_i$.
\item Suppose the above cases do not hold and $\varphi_i$ is defined
on $m_{i-1}00$ and $\varphi_i(m_{i-1}00) = \varphi(x)$ for some
$x > m_i11$. Let $y$ be equal to $x$ without its last two
characters, and set $m_i = \shift(y,1)$. Let
$A_i = A_{i-1} \cup\{m_{i-1}01\} \cup \{z11 \condition m_i < z <
y\} \cup \{w\}$, where $w$ is some element in
$\{y00,y01,y10,y11\} - \{x\}$. Note that $m_{i-1}00$ and $x$ are
both in $\overline{A_i}$ and lexicographically less than $m_i00$,
and take the same value under $\varphi_i$.
\end{enumerate}

Finally, let $A = \bigcup_{i \geq 0} A_i$. Notice that stage $i$
adds to $A$ only elements that are lexicographically greater than or
equal to $m_{i-1}00$, so if $x < m_i00$ and $x \in \overline{A_i}$,
then $x \in \overline{A}$. In case 1, we see that $\varphi_i$ fails
to be surjective (i.e., onto $\sigmastar$) when restricted to
$\overline{A}$, since there is no $x < m_i00$ in $\overline{A}$
mapping to $\varphi_i(m_{i-1}00)$, and also no $x \geq m_i0$ mapping
to $\varphi_i(m_{i-1}00)$.  In case 2, either $\varphi_i$ does not
halt on $m_i00 \in \overline{A}$ or there are two elements in
$\overline{A}$ that take the same value under $\varphi_i$.  In cases
3 and 4, we see that there are two elements in $\overline{A}$ that
take the same value under $\varphi_i$. Thus in all cases $\varphi_i$
fails to compress $\overline{A}$ and so $\overline{A}$ is not
$\fre$\tcomp{}.
\end{proof}

\begin{proof}[Proof of Theorem~\ref{green}]
Let $C$ be a set such that $A = C0 \cup \sigmastar1$ is not
\fprc{}. Such a set can be constructed using a similar method to
those of Theorems \ref{red}, \ref{blue}, and \ref{green}. Then
$B = C1 \cup \sigmastar0$, $A' = C00 \cup \sigmastar1$, and
$B' = C10 \cup \sigmastar1$ are all also not \fprc{}, since clearly
they are all recursively isomorphic. Note that
$A \cup B = \sigmastar$ and $A' \cap B' = \sigmastar1$, both of
which are \spr{}.
\end{proof}

\begin{proof}[Proof of Theorem~\ref{t:cpo}]
Let $A$ by any language that is not \wpr{} and whose complement is
not \wpr{}. An example of such a set is
$A=\{x000\condition x\in \Sigma^*\}\cup\{x001\condition x\in
B\}\cup\{x010\condition x\in \Sigma^*\}\cup\{x100\condition x\in
B\}$, where $B$ is any undecidable set. This set is not even
\frr{}. Note $x\in B$ if and only if
$\rank_A(x010)-\rank_A(x000)>1$, so if $A$ were \frr{}, we could
decide $B$. Similarly, $x\notin B$ if and only if
$\rank_{\overline{A}}(x101)- \rank_{\overline{A}}(x011)>1$, so if
$\overline{A}$ were \frr{}, $B$ would be decidable, but this is a
contradiction.

Then $A \join \overline{A}=A0\cup\overline{A}1$ is \wpr{}. It can be
ranked by any function mapping $x0$ and $x1$ to
$\rank_{\sigmastar}(x)$.
\end{proof}

\begin{proof}[Proof of Theorem~\ref{t:tyef}]
Let $A$ be some undecidable set. Let $A'=A\join \overline{A}$.  Then
$A'$ is \wpr{} by any function mapping $x0$ and $x1$ to
$\rank_{\sigmastar}(x)$. Now let $B=\sigmastar \join A'$. Then $B$
is the join of two \wpr{} sets. Suppose $B$ were $\p$\trank, then we
can query $\rank_B(x0)$ for all strings $x$. If
$\rank_B(x0)+2 = \rank_B(\shift(x,1)0)$, we know that $x1 \in B$,
and thus $x \in A'$. Otherwise, $x1 \notin B$ so $x1 \notin
A'$. Since we can test membership in $A'$, we can test membership of
$x$ in $A$ by asking whether $x0 \in A'$. This is a contradiction as
$A$ was assumed undecidable; thus $B$ cannot be $\p$\trank.
\end{proof}

\begin{proof}[Proof of Theorem~\ref{t:221}]
Let $A$ and $B$ be \spr{}. The rank of $x0$ in $A\join B$ is
$\rank_A(x)+\rank_B(\shift(x,-1))$, and the rank of $x1$ is
$\rank_A(x)+\rank_B(x)$. The rank of $\emptystring$ is 0. All of
these values can clearly be computed in polynomial time so
$A\join B$ is \spr{}.
\end{proof}

\begin{proof}[Proof of Theorem~\ref{t:903}]
Suppose \pr{} is closed under complement. Then \pr{} is equal to
\spr{}, so \pr{} is closed under join.

Now suppose \pr{} is closed under join. Let set $A$ be \pr{} by
ranking function $h$. Let $X=\sigmastar\join A$. Then $X$ is the
join of two \pr{} sets and thus is \pr{} by some ranking function
$f$. The ranking function for $\overline{A}$ does the following.
Given $x$, if $h(x)$ returns a rank (rather than an indication that
$x\notin A$) then return an indication that $x \notin \overline
A$. Otherwise let $y = \shift(x,1)$ and return
$2\rank_{\sigmastar}(x)+1-f(y0)$. There are a total of
$2\rank_\sigmastar(x)+2$ strings lexicographically less than or
equal to $y0$ in $\sigmastar$.  All those missing in $X$ correspond
to either $\emptystring$ or strings not in $A$ that are strictly
less than $y$.  Since $y0\in X$, we know that $f(y0)$ of these are
in $X$. The rest are in
$\sigmastar- X=\{x1\condition x\in \overline{A}\}\cup \{\epsilon\}$.
Thus the number of strings in $\overline{A}$ below $x$ is
$2\rank_\sigmastar(x)+1-f(y0)$. Thus $\overline{A}$ is \pr{}, so
\pr{} is closed under complement.
\end{proof}
\begin{proof}[Proof of Theorem~\ref{t:2243}]
Let $A$ and $B$ be two \pc{} sets. Let $f$ and $g$ be the
compression functions for $A$ and $B$ respectively. Let
$h(x0) = f(x)0$ and $h(x1) = \shift(f(x)0,-1)$. Now $h$ is a
compression function for $A \join B$ since the image of $h$
restricted to $A0$ is $\sigmastar0$, and the image of $h$ restricted
to $B1$ is $\sigmastar1 \cup \{\emptystring\}$. Each of $A0$ and
$B0$ maps injectively because $f$ and $g$ are compression functions,
and together they map injectively on all of $A \join B$ to all of
$\sigmastar$. The function $h$ is clearly polynomial time, and so
$A \join B$ is \pc{}.
\end{proof}
\begin{proof}[Proof of Theorem~\ref{t:rft:delta}]~
For the first part of this theorem, let $f$ be an
$\frec$-ranking function for $A$. Since $B_1$ and $B_2$ are
recursive, their ranking functions $\rank_{B_1}$ and $\rank_{B_2}$
are in $\frec$.
Our $\frec$ ranking function for $A \bigtriangleup (B_1 \cup B_2)$
is $f'(x) = f(x) + \rank_{B_2}(x) - \rank_{B_1}(x)$.  This directly
accounts for the additions and deletions done by $B_1$ and $B_2$.

We now prove the second part of the theorem.
The statement is clearly true if $B_1$ is finite, even in the case
that $A-B_1$ does not contain an infinite \re{} subset (as long as
$A-B_1$ is still infinite).  This is because the image of $A-B_1$
under a compression function for $A$ is cofinite, and cofinite sets
are compressible.  Thus composing a compression function for the
cofinite image of $A-B_1$ with a compression function for $A$, we
obtain a compression function for $A-B_1$.

So from this point on we assume that $B_1$ is infinite. Let $h$ be an
$\frec$-compression function for $A$. By the hypothesis of the theorem,
there is an infinite \re{} subset of $A-B_1$, call it $C$. Since every
infinite \re{} set contains an infinite recursive subset, let $B_2
\subseteq C$ be infinite and recursive. Let $b_1 < b_2 < b_3 < \cdots$ be
the elements in $B_1$, and let $c_1 < c_2 < c_3 < \cdots$ be the elements
in $B_2$. Consider the following function.

\begin{equation*}
g(x) =
\begin{cases}
x & \text{if}\ x \not\in B_1 \cup B_2,\\
\emptystring & \text{if}\ x \in B_1,\\
b_{\lceil i/2 \rceil} & \text{if}\ x = c_i \text{ and }i\text{ is odd, and}\\
c_{ i/2 } & \text{if}\ x = c_i \text{ and }i\text{ is even}.
\end{cases}
\end{equation*}
Let $f(x) = h(g(x))$. We claim that $f$ is a compression function for
$A-B_1$. We do this by showing $g$ is a compression function for $A-B_1$
onto $A$, since we already know that $h$ compresses $A$ to $\sigmastar$.
See that $g$ is the identity on $A-(B_1 \cup B_2)$.
See also that $g(B_2) = B_1 \cup B_2$ injectively and surjectively.
Since $(A-B_1)-B_2$ and $B_2$ are disjoint and have disjoint images, and since $g$ is injective and surjective on both these domains onto their respective images, it follows that $g$ is injective and surjective on $A-B_1$ to the image
$g((A-B_1)-B_2) \cup g(B_2) = A$.
Thus $g$ is a compression function for $A-B_1$ to $A$, and $h$ is a compression function for $A$ to $\sigmastar$, so $f$ is a compression function for $A-B_1$ to $\sigmastar$. In other words, $A-B_1$ is $\frec{}$\tcomp{}.

We now prove the third part of the theorem. Let $f$ be  an
$\frec$-compression function for $A$.
If $B_2$ is finite,
our $\frec$ compression function for $A \cup B_2$ is
$f'(x) = \shift(f(x),\|B_2\|)$ for $x \not \in B_2$ and
$f'(x) = \shift(\emptystring,\rank_{B_2}(x))$ for $x \in B_2$.

On other hand, if $B_2$ is infinite, let $g$ be an $\frec$
compression function for $B_2$, e.g., $g$ can be taken to be (recall
that $B_2$ is recursive) defined by $g(x)$ being the
$\max(\rank_{B_2}(x), 1)$-st string in $\sigmastar$.
We define $f'(x)$ as follows.  (Recall that for us $\Sigma$ is
always fixed as being $\{0,1\}$.)
If $x \not\in B_2$ then $f'(x) = 1f(x)$ (i.e., $f(x)$ prefixed with
a one).  If $x \in B_2$ and $g(x) = \emptystring$ then
$f'(x) = \emptystring$.  And, finally, if $x \in B_2$ and
$g(x) \neq \emptystring$ then $f'(x) = 0\shift(g(x),-1)$.  (The
shift-by-one treatment of the $x\not\in B$ case is because we must
ensure that $\emptystring$ is mapped to by some string in
$A\cup B_2$.) Now, $f'$ maps $A\cup B$ bijectively onto $\Sigma^*$, so $f'$ is an $\frec{}$ compression function for $A\cup B$, so $A\cup B$ is $\frec{}$\tcomp{}.
\end{proof}

\section{Appendix}\label{appendix:B:underdeveloped-parts}
\subsection{Relativization}

The results of
\cite{hem-rub:jtoappear-with-tr-pointer:rft-compression1} all relativize in a straightforward
manner. In this section, we include a few examples. This justifies our
limitation to $\frec$ and $\fre$:
By relativization, we get analogous results about more powerful
function classes, such as 
${\rm F}_{\Delta_{2}}$.\footnote{${\rm F}_{\Delta_{i+1}}$ will 
denote the class of total functions
computed by Turing machines given access to a $\Sigma_i$ oracle.
Equivalently, ${\rm F}_{\Delta_{i+1}} = (\frec)^{\Sigma_i}$.  Note
that $\frec = {\rm F}_{\Delta_{1}}$.  The class of partial functions
computed by Turing machines given access to a $\Sigma_i$ oracle will
be denoted $(\fre)^{\Sigma_i}$ or simply as
${\rm F}_\mathrm{PR}^{\Sigma_i}$.}

\begin{theorem}
For each $ i \geq 1$,
$\Delta_{i} = \Sigma_{i} \cap {\rm F}_{\Delta_i}$\tcomp$'$.
\end{theorem}

\begin{proof}
Relativization of
\cite[Theorem~5.3]{%
  hem-rub:t6OutByJournal:rft-compression1}
(see
also~\cite{gol-hem-kun:j:address,hem-rub:jtoappear-with-tr-pointer:rft-compression1}).
\end{proof}

Since, for $i \geq 1$ ${\rm F}_{\Delta_{i}} \supseteq \frec$, we get
the following easy corollary.

\begin{corollary}
For each $i \geq 1$, $\Sigma_{i} \cap
\frec$\tcomp$'\subseteq \Delta_i$.
\end{corollary}

\begin{theorem}
For each $i \geq 1$,
$\Pi_{i} \cap {\rm F}_{\Delta_{i}}\trank = \Pi_{i} \cap {\rm
  F}_\mathrm{PR}^{\Sigma_{i-1}}\trank$.
\end{theorem}

\begin{proof}
Relativization of
\cite[Theorem~4.6]{%
  hem-rub:t6OutByJournal:rft-compression1}
(see
also~\cite{hem-rub:jtoappear-with-tr-pointer:rft-compression1}).

\end{proof}

\subsection{Compressibility, Honesty, and Selectivity}

If we restrict our attention to honest functions, we can prove some
very clean results. There is a little subtlety here, since there are
many nonequivalent definitions of honesty. We use the following:

A (possibly partial) function $f$ is honest on $B$ if there is a
recursive function $g: \mathbb{N} \rightarrow \mathbb{N}$ such that
for any $x \in \domain(f) \cap B$, $g(|f(x)|) \geq |x|$. If $f$ is
honest on $\sigmastar$, we say $f$ is honest.

This gives two potential ways to define honest compressibility. For a
given set $A$, we can require the compression function to be honest on
$\sigmastar$, or only on $A$.  We call the former notion
honestly-$F$\tcomp\ , and the latter honestly-on-A-$F$\tcomp.  The
following theorem asserts that these two notions are equivalent for
$\fre$ and $\frec$ functions.

\begin{theorem}\label{t:115.1}
\begin{enumerate}
\item For each set $A$, $A$ is honestly-$\fre$\tcomp\ if and only if
$A$ is honestly-on-A-$\fre$\tcomp.
\item For each set $A$, $A$ is honestly-$\frec$\tcomp\ if and only
if $A$ is honestly-on-A-$\frec$\tcomp.
\end{enumerate}
\end{theorem}
\begin{proof}
The ``only if'' direction is trivial, since every honest on
$\sigmastar$ function is honest on $A$.

For the ``if'' direction, let $f$ be an honest on $A$ compression
function for $A$, and let $g$ be a recursive honesty-bound function
for $f$. Define $f'$ as follows, where $\domain(f') = \domain(f)$.
Over the domain of $f$, if $g(|f(x)|) \geq |x|$, then
$f'(x) = f(x)$.  If $g(|f(x)|) < |x|$, let $f'(x) = x$. Since $f$
was honest on $A$, for any $x \in A$, $f'(x) = f(x)$.  Thus $f'$ is
still a compression function for $A$.  The recursive function
$g'(n) = \max(g(n),n)$ satisfies, for all $x \in \domain(f')$,
$g'(|f'(x)|) \geq |x|$, and thus proves that $f'$ is honest (on
$\sigmastar$).

Note that when $f$ is recursive, so is $f'$, giving us the second
part of the theorem.
\end{proof}

The following proof uses $F$-selectivity, which was very rarely
useful. A set $A$ is $F$-selective if there is a function $f \in F$ of
two arguments such that the following hold:
\begin{enumerate}
\item For any $x,y \in \sigmastar$, either $f(x,y) = x$ or
$f(x,y) = y$.
\item If $x \in A$ or $y \in A$, $f(x,y) \in A$.
\end{enumerate}

Intuitively, $f$ selects the ``more likely'' of its two inputs. When
$x,y \not \in A$, or $x,y \in A$, $f$ can choose either input. It's
only restricted when one input is in $A$, and the other is not. Both
$\frec$-selectivity and honestly-$\frec$-compressibility are fairly
strong claims. Only the infinite recursive sets satisfy both.  Let
INFINITE denote the infinite sets over the alphabet
$\Sigma$.

\begin{theorem}\label{t:114.1}
$\frec$-selective $\cap$ honestly-$\frec$\tcomp\
$ = \rec \cap \INFINITE$.
\end{theorem}
\begin{proof}
Every infinite recursive set is easily $\frec$-selective and
honestly-$\frec$\tcomp\, giving the $\supseteq$ inclusion.

For the $\subseteq$ inclusion, let $A$ be honestly-$\frec$\tcomp\ by
$f$, with honesty bound $g$, and let $h$ be a $\frec$ selector
function for $A$. Then, for any $z$, by the definition of
compressibility and honesty:
$$ \| \{ w \condition f(w) = z \wedge |w| \leq g(|z|) \} \cap A \| = 1. $$

So define the finite set
$Q_z = \{ w \condition f(w) = z \wedge |w| \leq g(|z|) \}$. We know
this set contains exactly one element of $A$, and this will allow us
to decide $A$.

On input $x$, compute $f(x)$ and $Q_{f(x)}$. Then use the selector
function to find the unique element $y \in Q_{f(x)}$ such that for any
$z \in Q_{f(x)}$, $h(y,z) = y$. Such a $y$ exists and is unique
because there is exactly one element of $A$ in $Q_{f(x)}$.

If $x = y$, then $x \in A$. Otherwise, $x \not \in A$, so $A$ is
recursive. Since $A$ was compressible, it is infinite as well.
\end{proof}

In fact, honestly-$\frec$\tcomp\ is much stronger than
$\frec$\tcomp. While all $\core{}$ cylinders
are $\frec$\tcomp\ (see~\cite{hem-rub:jtoappear-with-tr-pointer:rft-compression1}),
no set in $\rm \re{} - REC$ is honestly-$\frec$\tcomp.  This was first
stated, without proof, in the conclusion section
of~\cite{gol-hem-kun:j:address}.

\begin{theorem}[See~\cite{gol-hem-kun:j:address}]
honestly-$\frec\mbox{\tcomp} \cap \core = \rec \cap \INFINITE$.
\end{theorem}
\begin{proof}
Every infinite recursive set is easily $\core{}$ and
honestly-$\frec$\tcomp\, giving the $\supseteq$ inclusion.

For the $\subseteq$ inclusion, let $A$ be $\core{}$ and
honestly-$\frec$\tcomp\ by a compression function $f$. Let $M$
accept $\overline{A}$ Define the sets $Q_z$ from the proof of
Theorem~\ref{t:114.1}.

Then for any input $x$, compute $f(x)$ and $Q_{f(x)}$. Then dovetail
applying $M$ to each element of $Q_{f(x)}$ until only one
remains. If the remaining element is $x$, then $x \in A$. Otherwise,
$x \not \in A$, so $A$ is recursive and infinite.
\end{proof}

This next group of theorems builds to a result that if $A$ is
nonrecursive, $\frec$-selective and $\frec$\tcomp\, then
$\overline{A}$ has an infinite \re{} subset. Since $\frec$-selectivity
is such a strong assumption, this theorem is of limited use. However,
the arguments used to show it may prove useful in the proof of other
claims.

\begin{theorem}\label{t:134.2}
If $A$ is $\frec$\tcomp\ via $f$ and $f(\overline{A})$ is finite,
then $A$ is recursive.
\end{theorem}
\begin{proof}
Using the definition of compressibility,
$L = \{ x \in A \condition f(x) \in f(\overline{A}) \}$ is finite.
By the assumptions of the theorem, so is $f(\overline{A})$. But,
$x \in A$ if and only if
$x \in L \lor f(x) \not \in f(\overline{A})$.  Since both of these
sets are finite, this condition is recursive, and so is $A$.
\end{proof}

Now we consider the case where $f(\overline{A})$ is infinite.
\begin{theorem}\label{t:135.1}
If $A$ if $\frec$\tcomp\ via $f$ and $f(\overline{A})$ is infinite,
then there is an infinite \re{} set
$B_A = \{ (p_1,q_1), (p_2, q_2), ... \}$ such that no string appears
in more than one pair and each pair contains at least one element of
$\overline{A}$.
\footnote{
In the theorem and the proof (and similarly regarding the proof of
the corollary) we should, to be formally correct, define and work
with the one-dimensional set
$\{ \langle a,b \rangle \condition (a,b) \in B_A\}$, where
$\langle \cdot , \cdot \rangle$ is a nice, standard pairing
function; but let us consider that implicit.
}
\end{theorem}
\begin{proof}
We describe a machine that enumerates the desired set.

Initialize $Q = \emptyset$. Begin running
$f(\epsilon), f(0), f(1), ...$ in sequence. Since $f(\overline{A})$
is infinite, there will be two strings $x,y \not \in Q$ where
$f(x) = f(y)$. Enumerate $(x, y)$, and add both $x$ and $y$ to
$Q$. The enumerated set will have the desired properties.
\end{proof}

\begin{corollary}\label{t:136.1}
If $A$ is nonrecursive, $\frec$-selective, and $\frec$\tcomp, then
$\overline{A}$ has an infinite \re{} subset.
\end{corollary}
\begin{proof}
Create the set from Theorem~\ref{t:135.1}, and apply the
$\frec$-selector to each pair. If the selector chooses $p_i$, then
$q_i \in \overline{A}$, and vice versa. So we can enumerate an
infinite $\re{}$ subset of $\overline{A}$ by enumerating the
elements not chosen by the selector.
\end{proof}

\end{document}